\documentclass[a4paper,8pt,twocolumn,preprint]{article}

\usepackage{booktabs} 
\usepackage{ifpdf}
\usepackage{cite}
\usepackage{amsthm}
\usepackage{array}
\usepackage{mdwmath}
\usepackage{mdwtab}
\def\bSig\mathbf{\Sigma}

\newtheorem{lemma}{Lemma}
\usepackage{algorithm}
\usepackage{amssymb}

\usepackage{caption}
\usepackage[noend]{algpseudocode}
\usepackage[utf8]{inputenc} 
\usepackage[T1]{fontenc}    
\usepackage{amsmath}
\usepackage{amsfonts}   
\usepackage{tikz}
\usepackage{amsopn}
\usepackage{bbding}
\usepackage{verbatim}
\usepackage{natbib}
\usepackage{epstopdf}
 \usepackage{lipsum}
\usepackage{graphicx}
\usetikzlibrary{calc,positioning}
\usepackage{etoolbox}
\usepackage{subcaption}
\usepackage{stfloats}
\usepackage{url}
\usepackage{authblk}
\makeatletter
\def\BState{\State\hskip-\ALG@thistlm}
\makeatother
\usetikzlibrary{calc,trees,positioning,arrows,chains,shapes.geometric,%
	decorations.pathreplacing,decorations.pathmorphing,shapes,%
	matrix,shapes.symbols}

\algdef{SE}[SUBALG]{Indent}{EndIndent}{}{\algorithmicend\ }%
\algtext*{Indent}
\algtext*{EndIndent}
\begin{document}
\pgfdeclarelayer{background}
\pgfdeclarelayer{foreground}
\pgfsetlayers{background,main,foreground}

\tikzset{%
	brace/.style = { decorate, decoration={brace, amplitude=5pt} },
	mbrace/.style = { decorate, decoration={brace, amplitude=5pt, mirror} },
	label/.style = { black, midway, scale=0.5, align=center },
	toplabel/.style = { label, above=.5em, anchor=south },
	leftlabel/.style = { label,rotate=-90,left=.5em,anchor=north },   
	bottomlabel/.style = { label, below=.5em, anchor=north },
	force/.style = {scale=0.9 },
	round/.style = { rounded corners=2mm },
	legend/.style = { right,scale=0.4 },
	nosep/.style = { inner sep=0pt },
	generation/.style = { anchor=base }
}
\tikzstyle{data}=[draw, fill=white!20, text width=3.3em, 
text centered, minimum height=2em]

\tikzstyle{dataexist}=[draw, fill=blue!20, text width=2em, 
text centered, minimum height=2em]

\tikzstyle{datanoexist}=[draw, fill=white!20, text width=2em, 
text centered, minimum height=2em,dashed]

\tikzstyle{n}=[draw, fill=white!20, text width=3.4em, 
text centered, minimum height=2.3em]
\tikzstyle{ann} = [above, text width=8em, text centered,scale=0.9]
\tikzstyle{wa} = [sensor, text width=10em, fill=white!20, 
minimum height=6em, rounded corners, drop shadow]
\tikzstyle{sc} = [sensor, text width=13em, fill=white!20, 
minimum height=10em, rounded corners, drop shadow]

\def\blockdist{2.3}
\def\edgedist{2.5}
\title{IS-ASGD: Accelerating Asynchronous SGD using Importance Sampling}

\author[1]{Fei Wang \thanks{bomber@sjtu.edu.cn}}
\author[2]{Weichen Li \thanks{weichenli@cmu.edu}}
\author[3]{Jason Ye \thanks{jason.ye.y@intel.com}}
\author[1]{Guihai Chen \thanks{gchen@cs.sjtu.edu.cn}}
\affil[1]{Department of Computer Science, Shanghai Jiao Tong University}
\affil[2]{Carnegie Mellon University}
\affil[3]{Intel Asia Pacific R\&D Ltd.}




\maketitle
\begin{abstract}
	Variance reduction (VR) techniques for convergence rate acceleration of stochastic gradient descent (SGD) algorithm have been developed with great efforts recently. VR's two variants, stochastic variance-reduced-gradient (SVRG-SGD) and importance sampling (IS-SGD) have achieved remarkable progresses. Meanwhile, asynchronous SGD (ASGD) is becoming more critical due to the ever-increasing scale of the optimization problems. The application of VR in ASGD to accelerate its convergence rate has therefore attracted much interest and SVRG-ASGDs are therefore proposed. However, we found that SVRG suffers dissatisfying performance in accelerating ASGD when the datasets are sparse and large-scale. In such case, SVRG-ASGD's iterative computation cost is magnitudes higher than ASGD which makes it very slow. On the other hand, IS achieves improved convergence rate with few extra computation cost and is invariant to the sparsity of dataset. This advantage makes it very suitable for the acceleration of ASGD for large-scale sparse datasets. In this paper we propose a novel IS-combined ASGD for effective convergence rate acceleration, namely, IS-ASGD. We theoretically prove the superior convergence bound of IS-ASGD. Experimental results also demonstrate our statements.
\end{abstract}
\section{Introduction}
For the empirical risk minimizations (ERM) problems, stochastic gradient descent (SGD) may be the most widely adopted solver algorithm. Let $w$ be the optimizer to be learned, denote 
\begin{equation}
f_i(w)=\phi_i(w)+\eta r(w),
\end{equation}
where $\phi_{i}$, $i\in\{1,2,...,n\}$ are vector functions that map $\mathbb{R}^d\to\mathbb{R}$ and $r(w)$ is the regularizer and $\eta$ is the regularization factor. This paper studies the following ERM optimization problem:
\begin{equation}
\label{opt}
\mathop{\mathrm{min}}_{w\in \mathbb{R}^d}F(w):=\frac{1}{n}\sum_{i=1}^{n}f_i(w).
\end{equation}
For SGD algorithm, $w$ is updated as:
\begin{equation}
\label{update}
w_{t+1} = w_t-\lambda\nabla f_{i_t}(w_t),
\end{equation} 
where ${i_t}\sim P$ means $i$ is drawn iteratively with respect to sampling probability distribution $P$ and $\lambda$ is the step-size. With the growing concurrency of hardwares, lock-free asynchronous SGD (ASGD) algorithms \cite{hogwild} have been developed to for speedup. With the improved speed and scalability, ASGDs quickly become indispensable and are de facto solvers for large-scale sparse optimizations. With the maturing of ASGDs, many following interests naturally shifted to the convergence acceleration techniques of them.
\subsection{Variance Reduction for Convergence Acceleration of ASGD}
It is commonly known that the variance of the stochastic gradient:
\begin{equation}
\mathbb{E}\left[\mathbb{V}(\nabla f_{i_t}(w_t)-\nabla F(w_t)\right],
\end{equation}
is one of the major reasons that slow down the convergence rate of SGD. The uses of variance reduction (VR) techniques to accelerate the iterative convergence rate of SGD have therefore attracted much interest recently. VR improves the iterative convergence rate by constructing variance-reduced gradient instead of using the original stochastic gradient directly for model update. 

One VR algorithm, stochastic variance-reduced-gradient (SVRG) \cite{SVRG} uses historical true-gradient and model snapshots to reduce the gradient variance. SVRG and its variants, e.g., SAGA \cite{Ghahramani} have been reported to be successful in accelerating the iterative convergence rate of SGD. Meanwhile, along with the intensively-studied SVRG-styled VR algorithms, another newly proposed VR technique, namely, importance sampling (IS) algorithm also achieves decreased stochastic gradient variance and improved convergence bound for SGD effectively by using non-uniform sampling (of the training samples) schemes as proposed in literatures \cite{p_zhao, Strohmer, Needell, csiba2016importance}.
\begin{algorithm}[t]
	\caption{Generic SVRG-styled ASGD}\label{general_g_VR}
	\begin{algorithmic}[1]
		\Procedure{SVRG-ASGD}{$T$}
		\State \textbf{Parallel} \textbf{do}
		\Indent \For{$t=0;t\not=T;t$++}
		\If{sync($t$)}
		\State $s=w_t$
		\State $\mu =\frac{1}{n}\sum_{i=0}^{n}\nabla f_i(s)$
		\EndIf
		\State $v_{t} = \nabla f_{i_t}(w_t)-\nabla f_{i_t}(s)+\mu$ \Comment{$i_t \sim P$}
		\State $w_{t+1} = w_t-\lambda v_t$\Comment{$w_t$ is the global model.}
		\EndFor
		\EndIndent
		\State \textbf{End Parallel}
		\State \textbf{return} 
		\EndProcedure
	\end{algorithmic}
\end{algorithm}

Recently, with the success of VR techniques, combining ASGD with VR to further improve its convergence rate shows practical significance and has therefore been studied, and several related works were proposed. Interestingly, we found that all these works by far are based on SVRG-styled ASGD (SVRG-ASGD), e.g., \cite{Shen-Yi, DBLP:conf/aaai/HuoH17, DBLP:conf/aaai/MengCYWML17, DBLP:conf/aaai/LiuSC17,Reddi} while the research of IS-styled ASGD is still untouched. One possible reason is that SVRG-SGD's iterative convergence rate, i.e., iteration count as the x-axis of the convergence curve, is much higher than that of IS-SGD. However, in practical deployments, the absolute convergence rate, i.e., wall-clock as the x-axis of the convergence curve, holds the actual significance. Unfortunately, previous works validated SVRG-styled ASGD with small-scale and relative dense datasets, in which its drawbacks on the absolute convergence rate for large-scale sparse datasets are not revealed.

\subsection{Absolute Convergence Acceleration: Sparsity and Performance}
After intensive evaluations of the existing SVRG-ASGD algorithms, we found that its absolute convergence rate is severely limited when dealing with large-scale sparse datasets, which is, unfortunately, de facto type of datasets that ASGDs are supposed to work with. See Algorithm 1 for the generic scheme of SVRG-ASGD, two bottlenecking issues decrease its absolute convergence rate drastically. They are caused by the same reason, i.e., SVRG is intrinsically dense.
\vskip 0.05in
\noindent\textbf{Sparsity for Less Computation} \: As can be seen in line 7, for each iteration of SVRG-ASGD, two additional vector adds, i.e., $\nabla f_{i_t}(s)$ and $\mu$ are needed. Intuitively, this increases the computation cost up to two times. However, the actual increase of the computation cost can be extremely large. See Figure 1 for illustration, we should be noted that in large-scale sparse optimization problems where ASGD is applied, stochastic gradient $\nabla f_{i_t}(w_t)$ is actually very sparse (as shown in the three upper rows) and is thus index-compressed, i.e., only the non-zeros features are stored with their corresponding indices. The update of $w_t$ is thus proceeded in an index-compressed way and the add operation is actually executed very few times, e.g., $10^{-7}$ * $d$, comparing to the dimensionality $d$. For sparse datasets with dimensions in tens of millions (which is not rarely seen in modern optimizations), index-compression of the sparse gradient is the most efficient method for ASGDs.
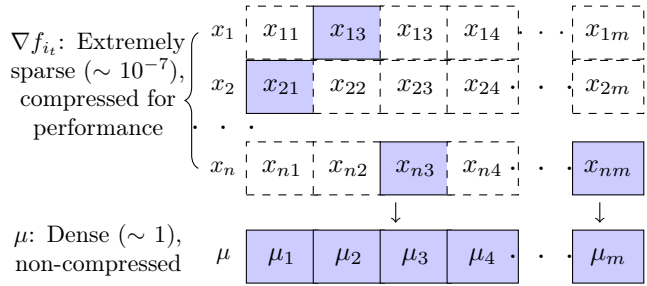
\begin{figure*}
	\centering
	\caption{Data Compression for Performance}
	\begin{tikzpicture}[x=2.2cm, y=1.8cm]
	\node at(-0.55,0.5) [force]   {$x_1$ };
	\node at(-0.55,0.1) [force]   {$x_2$ };
	\node at(-0.55,-0.5) [force]   {$x_n$ };
	
	\node at(-0.2, 0.5) (asr8) [datanoexist] {$x_{11}$};
	\node at(0.2, 0.5) (asr9) [dataexist] {$x_{13}$};
	\node at(0.6, 0.5) (asr10) [datanoexist] {$x_{13}$};
	\node at(1.0, 0.5) (asr11) [datanoexist] {$x_{14}$};
	\node at(1.35,0.5) [force]   {\large. . .};
	\node at(1.75, 0.5) (asr11) [datanoexist] {$x_{1m}$};
	
	\node at(-0.2, 0.1) (asr4) [dataexist] {$x_{21}$};
	\node at(0.2, 0.1) (asr5) [datanoexist] {$x_{22}$};
	\node at(0.6, 0.1) (asr6) [datanoexist] {$x_{23}$};
	\node at(1.0, 0.1) (asr7) [datanoexist] {$x_{24}$};
	\node at(1.35,0.1) [force]   {\LARGE. . .};
	\node at(1.75, 0.1) (asr11) [datanoexist] {$x_{2m}$};
	
	\node at(-0.55,-0.2) [force]   {\LARGE. . .};
	
	\node at(-0.2, -0.5) (asr4) [datanoexist] {$x_{n1}$};
	\node at(0.2, -0.5) (asr5) [datanoexist] {$x_{n2}$};
	\node at(0.6, -0.5) (asr6) [dataexist] {$x_{n3}$};
	\node at(1.0, -0.5) (asr7) [datanoexist] {$x_{n4}$};
	\node at(1.35,-0.5) [force]   {\LARGE. . .};
	\node at(1.75, -0.5) (asr11) [dataexist] {$x_{nm}$};
	
	
	\draw [mbrace] (-0.7,0.5) -- (-0.7,-0.5);
	\node at(-1.3,-0.35) [ann]   {$\nabla f_{i_t}$: Extremely sparse ($\sim 10^{-7}$), compressed for performance};
	\node at(-1.3,-1.35) [ann]   {$\mu$: Dense ($\sim 1$), non-compressed};
	\draw [->] (0.48,-0.75)   -- (0.48,-0.9);
	\draw [->] (1.7,-0.75)   -- (1.7,-0.9);
	\node at(-0.55,-1.15) [force]   {$\mu$};
	
	\node at(-0.2, -1.15) (asr4) [dataexist] {${\mu_1}$};
	\node at(0.2,-1.15) (asr5) [dataexist] {${\mu_2}$};
	\node at(0.6,-1.15) (asr6) [dataexist] {${\mu_3}$};
	\node at(1.0, -1.15) (asr7) [dataexist] {${\mu_4}$};
	\node at(1.35,-1.15) [force]   {\LARGE. . .};
	\node at(1.75, -1.15) (asr11) [dataexist] {${\mu_m}$};
	\end{tikzpicture}
\end{figure*}
However, for SVRG-ASGD, due to the involvement of the historical true gradient $\mu$, which is in fact a dense gradient with size $d$ as shown in the last row in Figure 1, the index-compressions of $\mu$ is meaningless. The update of $w_t$ has to be proceeded in the form of raw vector with full length $d$, which is typically five to seven magnitudes larger than the index-compressed stochastic gradient. Adding arrays of such large magnitude at each iteration is completely impractical. In consideration of the large data sample counts, the training can be extremely time-consuming even it actually needs less iterations (SVRG accelerates ASGD in iteration). In fact, for modern stochastic optimizations where sparse datasets with extra-high dimensionality are common, we found that performing SVRG-ASGD is computationally infeasible and often fails to complete training in a reasonable time due to the drastically increased computation cost, which is caused by the loss of sparsity. 

For large datasets, when the true-gradient is $10^{3}$ magnitudes higher than $\nabla f_i$, the absolute convergence rate of SVRG-ASGD shows large net decrease. Unfortunately, for the previous SVRG-ASGD works, the absolute convergence results were conducted on relative low dimensionality (around $10^{2}\sim10^{5}$) datasets; or using larger datasets ($10^{7}$) but the comparison is limited between SVRGs. 

We also find that the (only) public version of SVRG-ASGD\footnote{\url{https://github.com/CMU-ML-17-102/svrg.git}, committed by the author of~\cite{Reddi}} does not follow the proposed algorithm in its corresponding literature~\cite{Reddi}. The public version actually omit the addition of dense gradient $\mu$ at each iteration (in line 7) and only do it once at the end of each epoch by multiplying $\mu$ with $n$. It seems that the intention of this approximation is to avoid the expensive dense gradient operation at each iteration. Unfortunately, we found the convergence curve of this public version far from the literature version.

\noindent\textbf{Sparsity for Less Conflicts} \: It is commonly known that one fundamental assumption for ASGD is that the datasets are sufficiently sparse, otherwise the conflict updates for the global model would certainly raise the risk of non-convergence or inferior convergence curves. As a consequence, for SVRG-ASGD, the loss of sparsity due to the usage of dense true-gradient $\mu$ does not only increase the iterative time cost magnitudes higher but also increases the potentiality of conflict updates. Such conflicts weaken the benefits of using variance-reduced gradient, which can be deemed as another negative effect on the absolute convergence rate. 

Obviously, in oder to achieve the absolute converge rate acceleration of ASGD, a true-gradient-free VR algorithm has to be designed. Naturally, importance sampling as an elegant true-gradient-free VR algorithm comes to our consideration.

\subsection{IS-ASGD for Guaranteed Absolute Convergence Acceleration}
Clearly, when designing a VR algorithm to achieve absolute convergence acceleration for ASGD, we hope it not only remains a minimal increase of iterative time cost but also maintains low potentiality of conflict updates, which seems to be a difficult task. Fortunately, we notice that IS naturally suits in: it does not rely on the variance-reduced-gradient $v_t$ which makes it free from the true-gradient $\mu$ and thus the above mentioned performance-bottlenecking problems do not exist. In fact, IS can be implemented with no extra on-line computation by generating the sample sequences beforehand and let the computation threads iterate over the generated sequences, which leaves the computation kernel the same as ASGD. The calculation of the sampling distribution is typically fast which can actually be ignored comparing to the whole training time cost. That is, IS-ASGD is able to preserve almost the same iterative computation cost and low conflict updates with ASGD while achieving a higher iterative convergence rate. These are the key advantages for achieving a high absolute convergence rate acceleration. 

As mentioned above, since SVRG typically achieves much better performance on iterative convergence rate acceleration, research by far all focus on the SVRG-ASGD algorithms while the IS-styled ASGD algorithm is still left unstudied. We consider this missing field worthy to be researched due to its practical significance for high performance large-scale sparse optimizations and its novelty. Following this idea, we analyze and propose the algorithm that uses IS to accelerate the absolute convergence rate of ASGD effectively, i.e., IS-ASGD, as the novel contribution of this paper.

The rest of this paper is organized as follows. In section II we analyze the potential problems of applying IS in ASGD and propose IS-ASGD algorithm with detailed discussion. Section III is dedicated to the theoretical analysis of the convergence bound improvement of IS for ASGD. The in-depth evaluations of both iterative and absolute convergence results are provided in Section IV. Finally, the conclusion of this paper is given in Section V.
\section{Importance Sampling for Asynchronous SGD}
We first briefly introduce some key concepts of IS. Like most previous related literatures, we make the following necessary assumptions for the convergence analysis of the stochastic optimization problem studied in this paper.
\begin{itemize}
	\item \textbf{$\mu$-Convex}: $f_i$ is strongly convex with parameter $\mu$, that is:
	\begin{equation}
	\label{convexity}
	\langle x-y,\nabla f_i(x)-\nabla f_i(y) \rangle \ge \mu \|x-y\|_2^2, \, \forall x,y \in \mathbb{R}^d
	\end{equation}
	\item \textbf{$L_i$-Lipschitz}: Each $f_i$ is continuously differentiable and $\nabla f_i$ has Lipschitz constant $L_i$ w.r.t $\|\cdot\|_{2}$, i.e.,
	\begin{equation}
	\label{continuity}
	\|\nabla f_i(x)-\nabla f_i(y)\|_{2}\le L_i\|x-y\|_{2}, \, \forall x, y \in \mathbb{R}^d
	\end{equation}
\end{itemize}
where $\forall i \in \{1,2,...,N\}$.
\subsection{Importance Sampling}
Importance sampling reduces the gradient variance through a non-uniform sampling procedure instead of drawing sample uniformly. For conventional stochastic optimization algorithms, the sampling probability of $i$-th sample at $t$-th iteration, denoted by $p_i^t$, always equals to $1/N$ while in an IS scheme, $p_i^t$ is endowed with an importance factor $I_i^t$ and thus the $i$-th sample is sampled at $t$-th iteration with respect to a weighted probability:
\begin{equation}
p^{t}_i=I^t_i/N, \quad s.t. \sum_{i=1}^{N}p_i^t=1
\end{equation}
where $N$ is the number of training samples. With this non-uniform sampling procedure, to obtain an unbiased expectation, the update of $w$ is modified as:
\begin{equation}
\label{up}
w_{t+1} = w_t-\frac{\lambda}{np^t_{i_t}}\nabla f_{i_t}(w_t)
\end{equation}
where $i_t$ is drawn i.i.d w.r.t the weighted sampling probability distribution $P^t=\{p_i^t\}, \forall i\in \{1,2,...,N\}$.
\subsection{Importance Sampling for Variance Reduction} Recall the optimization problem in Equation \ref{opt}, using the analysis result from \cite{p_zhao}, we have the following lemma:
\begin{lemma}
	Let $\sigma^2=\mathbb{E}\|\nabla f_i(w_{\star})\|_2^2$ where $w_{\star}=\arg\underset{w}{\min}F(w)$. Set $\lambda \le \frac{1}{\mu}$, with the update scheme defined in Equation \ref{up}, the following inequality satisfy:
	\begin{equation}
	\begin{aligned}
	\mathbb{E}[F(w_{t+1})&-F(w_{\star})]\le \frac{1}{2\lambda}\mathbb{E}[\|w_{\star}-w_t\|_2^2-\|w_{\star}-w_{t+1}\|_2^2]\\ 
	&-\mu\mathbb{E}\|w_{\star}-w_t\|_2^2+\frac{\lambda_t}{\mu}\mathbb{E}\mathbb{V}\Big((np_{i_t}^t)^{-1}\nabla f_{i_t}(w_t)\Big)
	\end{aligned}
	\end{equation}
	where the variance is defined as:
	\begin{equation}
	\mathbb{V}\big((np^t_{i_t})^{-1}\nabla f_{i_t}(w_t)\big)=\mathbb{E}\|(np^t_{i_t})^{-1}\nabla f_{i_t}(w_t)-\nabla F(w_t)\|_2^2
	\end{equation}
	and the expectation is estimated w.r.t distribution $P^t$.
\end{lemma}
It is easy to verify that in order to minimize the gradient variance, the optimal sampling probability $p_i^t$ should be set as:
\begin{equation}
\label{precise}
p_i^t=\frac{\|\nabla f_{i}(w_t)\|_2}{\sum_{j=1}^N\|\nabla f_{j}(w_t)\|_2}, \qquad \forall i \in \{1,2,...,N\}.
\end{equation}
Obviously, such iteratively re-estimation of $P^t$ is completely impractical. The authors propose to use the supremum of $\|\nabla f_{i}(w_t)\|_2$ as an approximation. Since we have $L_i$-Lipschitz of $\nabla f_i$, by further assuming $\|w_t\| \le R$ for any $t$, we have $\|\nabla f_{i}(w_t)\|_2 \le RL_i$, i.e., $\sup\|\nabla f_{i}(w_t)\|_2=RL_i$. Thus the actual sampling probability of $p_i$ is calculated as:
\begin{equation}
\label{actual}
p_i=\frac{L_i}{\sum_{j=1}^{N}L_j}, \, \forall i \in\left\{1,2,...,N\right\}.
\end{equation} 
With such definition, $P$ needs no update and is used throughout the training procedure. The authors further prove that with Equation~\ref{actual}, IS accelerated SGD achieves a convergence bound as:
\begin{equation}
\frac{1}{T}\sum_{t=1}^{T}\mathbb{E}[F(w_t)-F(w_{\star})]\le \sqrt{\frac{\|w_{\star}-w_0\|_2^2}{\sigma}}\left(\frac{\sum_{i=1}^{n}L_i}{n}\right)\frac{1}{T},
\end{equation}
while for standard SGD solver that actually samples $x_i$ w.r.t uniform distribution, the convergence bound is:
\begin{equation}
\frac{1}{T}\sum_{t=1}^{T}\mathbb{E}[F(w_t)-F(w_{\star})]\le \sqrt{\frac{\|w_{\star}-w_0\|_2^2\sum_{i=1}^{n}(L_i^2)}{\sigma n}} \frac{1}{T},
\end{equation}
when $\lambda$ is set as $\sqrt{\sigma\|w_{\star}-w_0\|_2^2}/\left(\frac{\sum_{i=1}^{n}L_i}{n}\sqrt{T}\right)$. According to Cauchy-Schwarz inequality, we always have $\frac{(\sum_{i=1}^{n}L_i)^2}{n\sum_{i=1}^{n}L_i^2} \le 1$, which implies that IS does improve convergence bound. Denote
\begin{equation}
\psi=\frac{(\sum_{i=1}^{n}L_i)^2}{\sum_{i=1}^{n}(L_i^2)},
\end{equation}
we can conclude that the improvement gets larger when $\psi \ll n$. 
\begin{algorithm}[t]
	\caption{Practical Importance Sampling for SGD}
	\begin{algorithmic}[1]
		\Procedure{IS-SGD}{$T$}
		\State Construct $P=\{p_i\}_{i=1}^N$ According to Equation \ref{actual}
		\State Generate Sample Sequence $S$ w.r.t distribution $P$.
		\For{$i=0;i\not=T;i$++}
		\State $i_t=S[i]$
		\State $w_{t+1} = w_t-\frac{\lambda}{np_{i_t}}\nabla f_{i_t}(w_t)$
		\EndFor\label{IS-practical}
		\EndProcedure
	\end{algorithmic}
\end{algorithm}

For example, for L2-regularized SVM optimization problem with squared hinge loss, i.e., $f_i(w)=(\lfloor 1-y_iw^{T}x_i \rfloor_{+})^2+\frac{\lambda}{2}\|w\|_2^2$, where $x_i$ is the $i$-th sample and $y_i \in \{-1,+1\} $ is the corresponding label, $\|\nabla f_{i}(w)\|_{2}$ can be bounded as 
\begin{equation}
\|\nabla f_{i}(w)\|_{2} \le 2(1+\|x_{i}\|_{2}/\sqrt{\lambda})\|x_{i}\|_{2}+\sqrt{\lambda}
\end{equation}
The pseudo code of practical IS-SGD algorithm is shown in Algorithm 2. As can be seen that, the core procedure of IS is the construction of sampling distribution $P$. Once $P$ is constructed, IS-SGD works as same as SGD except that the training samples are selected w.r.t to weighted probability distribution $P$ and the step-size is adjusted with ${1}/{np_i}$. It is clear that IS-ASGD does not rely on the true gradient $\mu$, which makes it free from the bottlenecking issues that deteriorate the absolute convergence rate of SVRG-ASGD drastically. This means that IS as an effective VR technique is very suitable for ASGD solvers with large-scale sparse datasets.
\subsection{Importance Imbalance}
In most ASGD implementations that solve large-scale optimization problems, each training thread/process runs on its corresponding core/node and typically works on its local dataset for the sake of performance and scalability. For IS-ASGD, such data-segmentation brings in problem since each thread/process (indexed with $i$) can only calculate the sampling probability distribution $P_i$ based on its local dataset instead of the whole dataset, which leads to sub-optimal VR performance of IS for ASGD. See Figure~\ref{imbc} for illustration, assume we have two working cores and whole training dataset as:
\begin{equation}
\mathcal{D} = \{x_1, x_2, x_3, x_4\}
\end{equation}
with subset $\mathcal{D}_1=\{x_1, x_2\}$ located on core/node 1 while $\mathcal{D}_2=\{x_3, x_4\}$ located on core/node 2. Without loss of generality, we further assume their corresponding Lipschitz constants as $\{1, 2, 3, 4\}$. For comparison, in IS-SGD where the only training process works on the whole dataset, the probability distribution of being chosen as the training sample is $P=\{p_1=0.1, p_2=0.2, p_3=0.3, p_4=0.4\}$ while in IS-ASGD with local-data-training, the sampling probabilities are $P_1=\{p_1=0.33, p_2=0.67\}$ and $P_2=\{p_3=0.43, p_4=0.57\}$ respectively for each core/node. In global-data-training algorithm e.g., IS-SGD, $p_4$ is much larger (twice over) than $p_2$ while in IS-ASGD, $p_4$ is even smaller than $p_2$ which is a heavy distortion from the theoretical optimum.
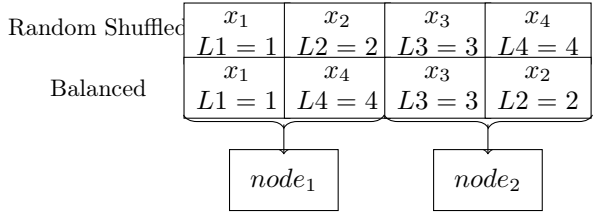
\begin{figure}[t]
	\centering
	\caption{Importance Balancing for ASGD}
	\begin{tikzpicture}[x=2.2cm, y=1.8cm]
	\node at(-0.63,0.55) [force]   {Random Shuffled};
	\node at(-0.63,0.1) [force]   {Balanced};
	
	\node at(0.2, 0.5) (asr8) [data] {$x_1$\\$L1=1$};
	\node at(0.8, 0.5) (asr9) [data] {$x_2$\\$L2=2$};
	\node at(1.4, 0.5) (asr10) [data] {$x_3$\\$L3=3$};
	\node at(2.0, 0.5) (asr11) [data] {$x_4$\\$L4=4$};
	
	\node at(0.2, 0.1) (asr4) [data] {$x_1$\\$L1=1$};
	\node at(0.8, 0.1) (asr5) [data] {$x_4$\\$L4=4$};
	\node at(1.4, 0.1) (asr6) [data] {$x_3$\\$L3=3$};
	\node at(2.0, 0.1) (asr7) [data] {$x_2$\\$L2=2$};
	
	
	\node at(0.48, -0.58) (n2) [n] {$node_1$};
	\node at(1.7, -0.58) (n4) [n] {$node_2$};
	
	\draw [mbrace] (-0.12,-0.1) -- (1.08,-0.1);
	\draw [mbrace] (1.08,-0.1) -- (2.325,-0.1);
	\draw [->] (0.48,-0.17)   -- (0.48,-0.39);
	\draw [->] (1.7,-0.17)   -- (1.7,-0.39);
	\end{tikzpicture}
	\label{imbc}
\end{figure}
\begin{algorithm}[t]
	\caption{Importance\_Balancing}\label{datareorg}
	\begin{algorithmic}[1]
		\Procedure{Importance\_Balancing}{$\mathcal{D}$}
		\State Calculate $L=\{L_i\} \quad \forall i\in\{1,2,...,N\}$
		\State $\mathcal{D}_s=$Get\_Sorted\_Indices$(\mathcal{D})$ w.r.t $L$
		\For{$i=0,idx=0;i<n/2;i$++}
		\State $\mathcal{D}_{r}[idx++]=\mathcal{D}_s[i]$
		\State $\mathcal{D}_{r}[idx++]=\mathcal{D}_s[n-1-i]$
		\EndFor
		\State \textbf{if} $n\%2$ \textbf{then}
		\State \ \ \ \ \ \ $\mathcal{D}_{r}[idx]=\mathcal{D}_s[n/2]$	
		\State Return $\mathcal{D}_{r}$, $L$
		\EndProcedure
	\end{algorithmic}
\end{algorithm}
\subsection{Importance Balancing for IS-ASGD}
To reduce such imbalance, a rearrangement of the dataset before dividing/dispatching data segments to its corresponding core/node should be considered. See the second row of Figure~\ref{imbc} for illustration, to achieve a balanced importance segmentation, we design a simple balancing algorithm as shown in Algorithm~\ref{datareorg}. As can be seen that this procedure generates rearranged dataset indices $\mathcal{D}_r$ by locating $\mathcal{D}_s[i]$ and $\mathcal{D}_s[n-1-i]$ index together sequentially. Denote $\Phi_a$ as the importance sum of core/node $a$:
\begin{equation}
\Phi_a=\sum_{i=1}^{N_a}L_{a_i}
\end{equation} 
where $L_{a_i}$ is the Lipschitz constant of the $i$-th data sample on core/node $a$, $N_a$ is the number of data samples on core/node $a$. According to Equation \ref{actual}, we have the sampling probability of $i$-th data sample on core/node $a$ as $p_{a_i}=\frac{L_{a_i}}{\Phi_a}, \forall i \in\left\{1,2,...,N_a\right\}$. It is easy to prove that by satisfying
\begin{equation}
\Phi_a = \Phi_b \quad \forall a,b \in\{1,...,num_T\}
\end{equation}
where $num_T$ is the number of cores/nodes, then the importance imbalance is eliminated. We call this dataset rearrangement procedure as importance balancing.
Obviously Algorithm~\ref{datareorg} does not guarantee to produce an equal-importance dataset segmentation. However, segmenting dataset into certain number (e.g. $num_T$) of equal-importance subsets is a typical NP-hard problem which can not be solved easily. We still use this simple head-tail sequential matching procedure since it is a fast approximation and generally works well in practice. 

Meanwhile, it has to be pointed out that if the distribution of the Lipschitz constants closes to uniform distribution or the dataset is sufficiently large, a random shuffling would work just fine for IS to perform VR since the risk of severe importance imbalance is low. We empirically define a metric $\rho$, which measures the potential of the imbalance to some extent:
\begin{equation}
\rho = \frac{\sum_{i=1}^{N}(L_i-\mu)^2}{N},
\end{equation}
where $\mu=\sum_{i=1}^{N}L_i/N$. A lower $\rho$ indicates lower potential of severe importance imbalance and vice versa. Accordingly, IS-ASGD is designed to perform importance balancing in an adaptive manner depending on $\rho$. The pseudo code of IS-ASGD is shown in Algorithm~\ref{IS-ASGD}, $\zeta$ is empirically set as $5^{-4}$.

\section{Convergence Analysis of IS-ASGD}
Among the many analysis of the convergence bound of ASGD, Horia et al. model ASGD as SGD with perturbed inputs i.e., the inconsistent state of the model is treated as true model with noise added. Comparing to other convergence analysis, this scheme is more general, compact and most importantly, makes the analysis of the effect of IS in ASGD relative simple. We first give a brief introduction of the perturbed iterate analysis which serves as the base of our analysis. For the ease of analysis, we presume that the dataset is perfectly importance-balanced, i.e., $\Phi_a = \Phi_b$, $\forall a,b \in\{0,1,...,num_T\}$, that is, IS achieves its theoretical convergence bound proved in previous literatures.
\begin{algorithm}[t]
	\caption{IS-ASGD Algorithm}\label{IS-ASGD}
	\begin{algorithmic}[1]
		\Procedure{IS-ASGD}{$num_T$, $T$, $\mathcal{D}$}
		\State Calculate $\rho$
		\If {$\rho \le \zeta$}
		\State $\mathcal{D}_r$, $L$=Importance\_Balancing($\mathcal{D}$)
		\Else
		\State $\mathcal{D}_r$, $L$=Random\_Shuffling($\mathcal{D}$)
		\EndIf
		\State {\bf Parallel do} with $num_T$
		\Indent
		\State $tid = $ GetTid()
		\State $\mathcal{D}_{tid}=\mathcal{D}_{r}[\frac{n*tid}{num_T}:\frac{n*(tid+1)}{num_T}]$
		\State $L_{tid}=L[\mathcal{D}_{tid}]$
		\State Calculate $P_{tid}$ w.r.t $L_{tid}$
		\State Generate Local Sample Sequence $S_{tid}$ w.r.t $P_{tid}$
		\For{$i=0;i\not=T;i$++}
		\State $i_t=\mathcal{D}_{tid}[S_{tid}[i]]$
		\State $w_{t+1} = w_t-\frac{\lambda}{np_{i_t}}\nabla f_{i_t}(w_t)$
		\EndFor
		\EndIndent
		\State \textbf{return}
		\EndProcedure
	\end{algorithmic}
\end{algorithm}
\subsection{Perturbed Iterate Analysis}
In perturbed iterate analysis \cite{mania2017perturbed}, the update of $w_t$ is modeled as:
\begin{equation}
w_{t+1} = w_t-\lambda \nabla f_{i_t}(w_t+\theta_t)
\end{equation}
where $\theta_t$ is the asynchrony error term caused by lock-free update at iteration $t$. Let $\hat{w_t}=w_t+\theta_t$. We have:
\begin{equation}
\begin{aligned}
\label{def}
\|w_{t+1}-w_{\star}\|^2_{2}=&\|w_t-\lambda \nabla f_{i_t}(\hat{w_t})-w_{\star}\|^2_{2}\\
=&\|w_t-w_{\star}\|_2^2-2\lambda \langle \hat{w_t}-w_{\star},\nabla f_{i_t}(\hat{w_t})\rangle+\\
&\lambda^2\|\nabla f_{i_t}(\hat{w_t})\|_2^2+2\lambda \langle \hat{w_t}-w_t,\nabla f_{i_t}(\hat{w_t})\rangle
\end{aligned}
\end{equation}
Recall the convexity assumption i.e., $f_i$ is strongly convex with parameter $\mu$, we have:
\begin{equation}
\begin{aligned}
\label{ineq}
\langle \hat{w_t}-w_{\star}, \nabla f_{i_t}(\hat{w_t}) \rangle \ge \mu \|\hat{w_t}-w_{\star}\|_2^2 \\
\mu\|\hat{w_t}-w_{\star}\|_2^2 \ge  \frac{\mu}{2} \|w_t-w_{\star}\|_2^2-\mu \|\hat{w_t}-w_t\|_2^2
\end{aligned}
\end{equation}
Denote by $\epsilon_t$ the relative error of $\hat{w_t}$, i.e, $\mathbb{E}\|\hat{w}_t-w_{\star}\|_2^2$. By substituting Equation \ref{ineq} back to \ref{def}, we obtain:
\begin{equation}
\begin{aligned}
\label{bound}
\epsilon_{t+1} \le (1-\lambda\mu)\epsilon_t&+\lambda^2\underbrace{\mathbb{E}\|\nabla f_{i_t}(\hat{w_t})\|_2^2}_{R_0^t}+2\lambda\mu\underbrace{\mathbb{E}\|\hat{w_t}-w_t\|_2^2}_{R_1^t}\\
&+2\lambda\underbrace{\mathbb{E}\langle\hat{w_t}-w_t,\nabla f_{i_t}(\hat{w_t})\rangle}_{R_2^t}
\end{aligned}
\end{equation}
Among the three labeled terms, notice that $R_0^t$ is a common term that exists in both SGD and ASGD while $R_1^t$ and $R_2^t$ are additional error terms introduced by the {\it inconsistency} of the model. $R_1^t$ reflects the difference between the true model and the perturbed (noise added) one, and $R_2^t$ measures the projection of such noise on the gradient of each iteration. Now that the convergence bound can be obtained once $R_0^t$, $R_1^t$ and $R_2^t$ are bounded. The authors first bound $\mathbb{E}\|\nabla f_{i_t}(\hat{w_t})\|_2\le M$, i.e., $R_0^t \le M^2$. Next, to bound $R_1^t$ and $R_2^t$, the concept of {\it conflict graph} is introduced as the following.
\paragraph{Conflict graph}
Denote by $c_i \subseteq \{j\}_{j=0}^d$ the set of feature index of data sample $x_i$, i.e., $j \in c_i$ only if the $j$-th feature is provided in $x_i$. In a conflict graph $G=\{e_{ij},v_i\}$, $i,j \in \{0, 1, ...,n\},\ i\ne j$, vertices $v_i$ and $v_j$ are connected with edge $e_{ij}$ if and only if $c_i \cap c_j \ne \varnothing$. Further define two factors that reflect the extent of conflict update: 
\begin{itemize}
	\item {\bf Delay parameter}, $\tau$, i.e., the maximum lag between when a gradient is computed and
	when it is applied to $w$. It is assumed that $\tau$ is linearly related to the concurrency.
	\item {\bf Conflict parameter}, $\bar{\Delta}$, which is the average degree of the conflict graph $G$, obviously, datasets with higher $\bar{\Delta}$ suffers severer extent of conflict updates and vice versa.
\end{itemize}
These two parameters measure the extent of inconsistency from two aspects. $\tau$ is set as the proxy of concurrency of ASGD which can be controlled by the users while $\bar{\Delta}$ measures the intrinsic potentials of conflict update of dataset which is irrelevant to the algorithm's settings. The authors prove that $R_1^t$ is bounded as $R_1^t \le \lambda^2M^2\Big(2\tau+8\tau^2\frac{\bar{\Delta}}{n}\Big)$ and $R_2^t$ bounded as $R_2^t \le 4\lambda M^2 \tau\frac{\bar{\Delta}}{n}$. Thus the recursion can be obtained by plugging $R_1^t$ and $R_2^t$ back to Equation \ref{bound}:
\begin{equation}
\label{asgd}
\begin{aligned}
\epsilon_{t+1} \le &\underbrace{(1-\lambda\mu)\epsilon_t+\lambda^2M^2}_{\xi}\\
&+\underbrace{\lambda^2M^2\Big(8\tau\frac{\bar{\Delta}}{n}+4\lambda\mu\tau+16\lambda\mu\tau^2\frac{\bar{\Delta}}{n}}_{\delta}\Big)
\end{aligned}
\end{equation}
\subsection{Bounding IS-ASGD}
Now the recursion of $\epsilon_{t}$ is divided into two parts, i.e., accurate SGD term $\xi$ and noise term $\delta$. With such scheme of modeling, the difficulty of the analysis caused by the inconsistency can be greatly simplified. We thus have the following lemma that bounds IS-ASGD.
\begin{lemma}
	\label{cvg}
	For IS-ASGD algorithm that follows the scheme of Algorithm~\ref{IS-ASGD}, by satisfying the convexity and continuity conditions in Equation \ref{convexity} and Equation \ref{continuity}. Denote by $\sigma$ the residual, i.e., $\mathbb{E}\|\nabla f_i(w_{\star})\|_2$, with a proper stepsize as $\lambda=\epsilon\mu/(2\epsilon\mu\sup L+2\sigma^2)$, the iteration steps $k$ which is sufficient to achieve $\mathbb{E}\|w_k-w_{\star}\|_2^2 \le \epsilon$, is defined as:
	\begin{equation}
	k= O(1) \log(\epsilon_0/\epsilon)\left(\frac{\bar{L}}{\mu}+\frac{\bar{L}}{\inf L}\frac{\sigma^2}{\mu^2\epsilon}\right)
	\end{equation}
	when $\tau$ is bounded as 
	\begin{equation}
	O\Big(\min\Big\{n/\bar{\Delta},\frac{\epsilon\mu\sup L+\sigma^2}{\epsilon\mu^2}\Big\}\Big)
	\end{equation}	
	where $\epsilon_0:=\max_{0\le t\le T}\mathbb{E}\|\hat{w_t}-w_{\star}\|_2^2$.
\end{lemma}
\begin{proof}
	Using the analysis of \cite{Needell}, we know that for $\xi$, the convergence bound of SGD is obtained as
	\begin{equation}
	\label{bsgd}
	k = 2\log(\epsilon_0/\epsilon)\left(\frac{\sup L}{\mu}+\frac{\sigma^2}{\mu^2\epsilon}\right)
	\end{equation}
	when $\lambda = \epsilon\mu/(2\epsilon\mu\sup L+2\sigma^2)$ and the convergence bound of $\xi$ is further reduced from supremum dependence of $L$ to average dependence through the application of IS, i.e., 
	\begin{equation}
	\label{issgd}
	k = 2\log(\epsilon_0/\epsilon)\left(\frac{\bar{L}}{\mu}+\frac{\bar{L}}{\inf L}\frac{\sigma^2}{\mu^2\epsilon}\right)
	\end{equation}
	With the accurate SGD term $\xi$ bounded as Equation \ref{issgd}, we are left to bound the noise term $\delta$ as an order-wise constant in order to achieve nearly linear speedup of IS-SGD. 
	
	According to the definition of $p_i$ as shown in Equation \ref{actual}, we have $(np_i)^{-1}=\frac{\bar{L}}{L_i}$. Since $\nabla f_{i_t}$ is scaled with $(np_i)^{-1}$ in IS, $M$ is also scaled as $M_s:=(np_i)^{-1}M$. Since $\frac{\bar{L}}{L_i} \le \frac{\bar{L}}{\inf L}$, $M_s$ is thus bounded as:
	 \begin{equation}
	 \begin{aligned}
	M_s\le\frac{\bar{L}}{\inf L}M
	 \end{aligned}
	 \end{equation}
	  With this result, from Equation \ref{asgd}, it can be concluded that $\delta$ is bounded as an order-wise constant when the following conditions are satisfied:
	\begin{equation}
	\begin{aligned}
	\label{requirement}
	\tau\frac{\bar{\Delta}}{n}=O(1) \quad and \quad \tau\lambda\mu=O(1)
	\end{aligned}
	\end{equation}
	Considering that $\lambda$ is set as $\epsilon\mu/(2\epsilon\mu\sup L+2\sigma^2)$, Equation \ref{requirement} is thus satisfied by bounding $\tau$ as:
	\begin{equation}
	\label{bt}
	O\Big(\min\Big\{n/\bar{\Delta},\frac{\epsilon\mu\sup L+\sigma^2}{\epsilon\mu^2}\Big\}\Big)
	\end{equation}
	Thus the recursion of IS-ASGD is the same with IS-SGD plus an additional order-wise constant. We thus have the convergence bound of IS-ASGD as shown in Lemma 2.
\end{proof}
Obviously this bound inherits the superiority of IS-SGD over ASGD, and it shows that IS-ASGD achieves a nearly linear speedup of IS-SGD which is similar to the previous result in \cite{congfang} that shows SVRG-ASGD achieves nearly linear speedup of SVRG-SGD.

In brief, the key to the convergence bound analysis is the serialization of the asynchrony which divides the update scheme into two self-bounded terms, i.e., $\xi$ and $\delta$. Such separation makes the analysis much simpler, that is, IS decreases the convergence bound of $\xi$ as the same as in SGD while the two bounded error terms caused by the asynchrony, i.e., $R^t_1$, $R^t_2$, increase the convergence bound up to a constant when certain conditions are met.

\section{Experimental Results}
In order to make our evaluation representative and convincing, we conduct the evaluation based on the following configuration: 
\vskip 0.1cm
\noindent \textbf{Testbed}\hspace{0.2cm}Our testbed is a 2-sockets server with Intel XeonE5-2699V4 CPU which has 44 cores in total (with HyperThreading off) and 128G main memory.
\vskip 0.1cm
\noindent \textbf{Code Base}\hspace{0.2cm}We base our IS-ASGD code on well-validated open source version of ASGD algorithm\footnote{\url{http://i.stanford.edu/hazy/victor/Hogwild/}}. We also make our evaluation code of IS-ASGD publicly available\footnote{\url{https://github.com/FayW/IS-ASGD.git}}.
\vskip 0.1cm
\noindent \textbf{Datasets}\hspace{0.2cm}Evaluation datasets are from LibSVM\footnote{\url{https://www.csie.ntu.edu.tw/~cjlin/libsvmtools/datasets/}}. \textit{News20} (low dimensionality and relative dense) was used for the purpose of validation in previous SVRG-ASGD works. Such small dataset can not expose the bottlenecking performance problems as discussed in section 1. Yet we still select this dataset for comparison. For other three large-scale datasets, SVRG-ASGD fails to finish training in a reasonable time\footnote{For KDD datasets, SVRG-ASGD takes about 2 hours to finish 1 epoch with 44 threads.} and thus we show no comparison of it.

According to the empirical threshold we set for $\rho$ in section 2, \textit{News20} is importance-balanced while for other datasets, we use simple random shuffling for IS-ASGD.
\vskip 0.1cm
\begin{figure*}[t]
	\centering
	\subcaptionbox{JMLR\_News20, $\lambda = 0.5$}{\includegraphics[scale=0.33]{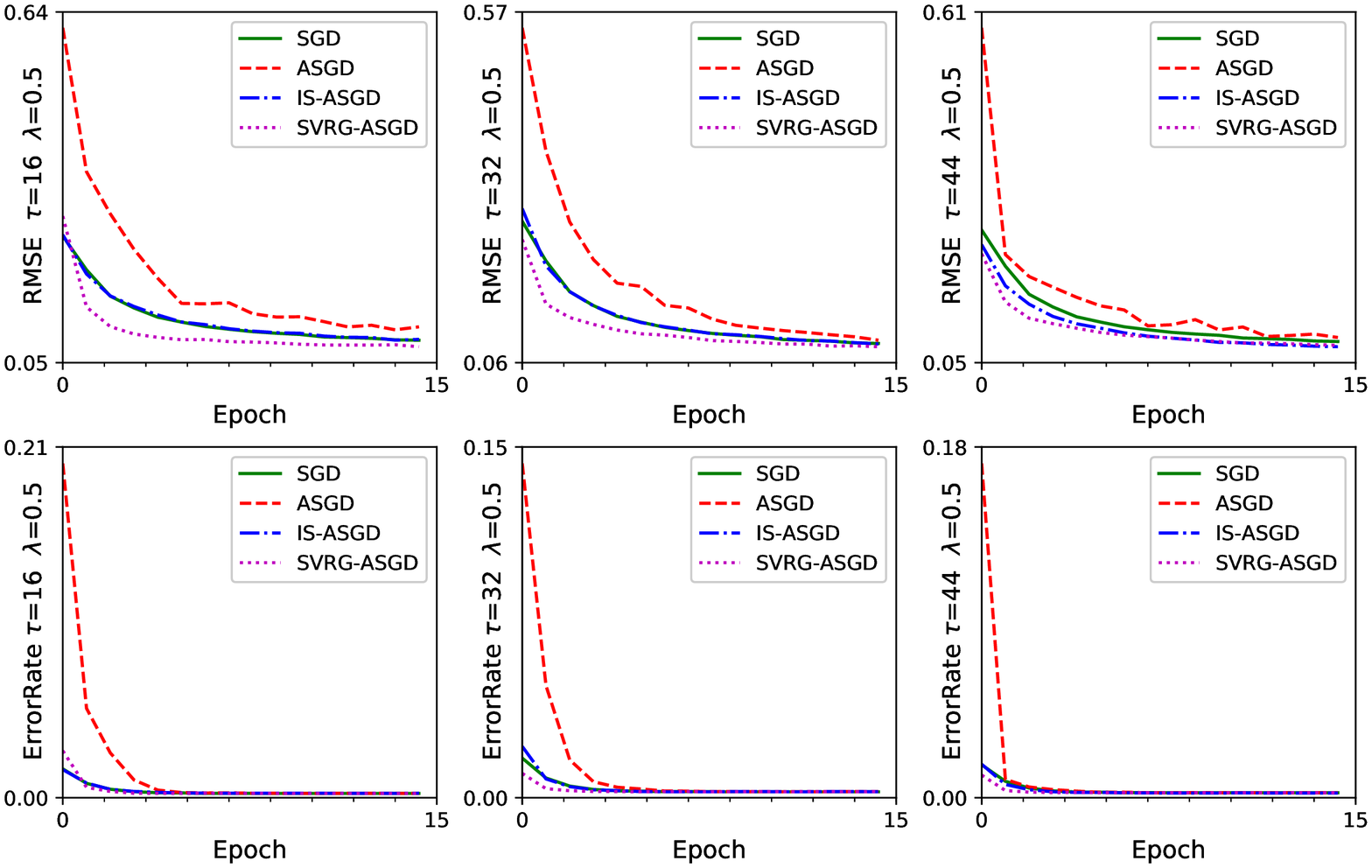}}%
	\subcaptionbox{KDD2010\_Alg., $\lambda = 0.5$}{\includegraphics[scale=0.33]{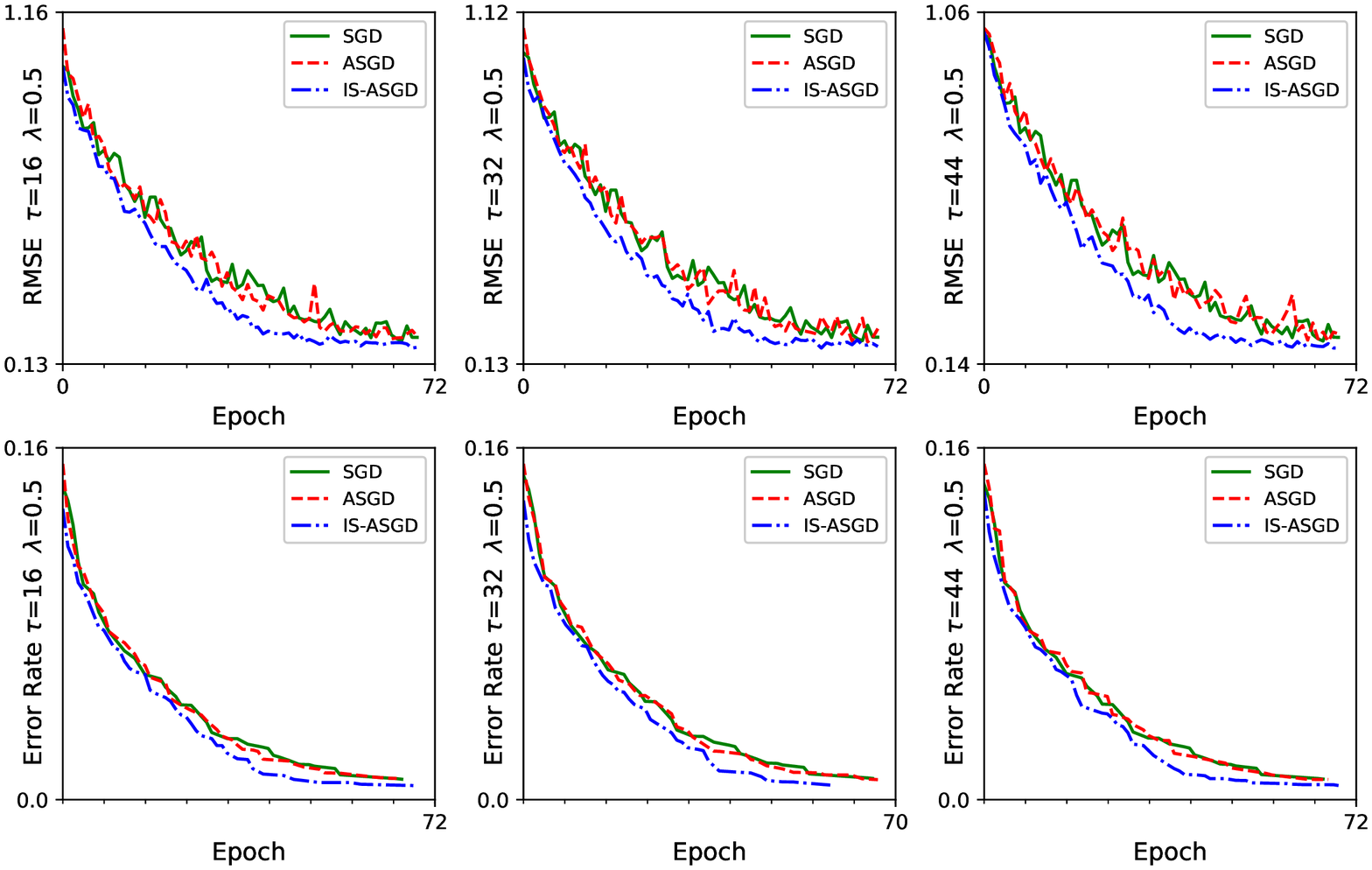}}
	
	\subcaptionbox{ICML\_URL, $\lambda = 0.05$}{\includegraphics[scale=0.33]{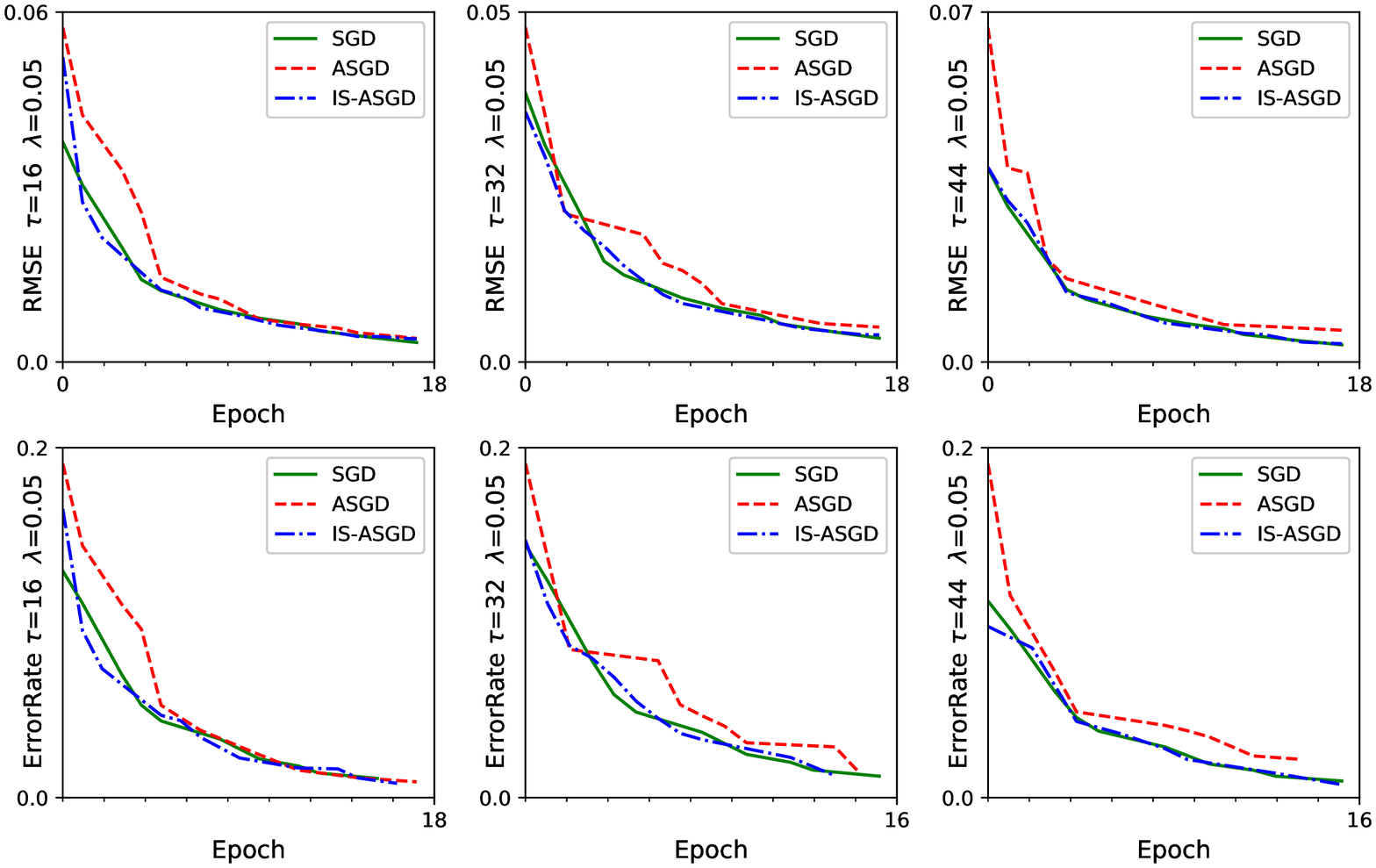}}%
	\subcaptionbox{KDD2010\_Bri., $\lambda = 0.5$}{\includegraphics[scale=0.33]{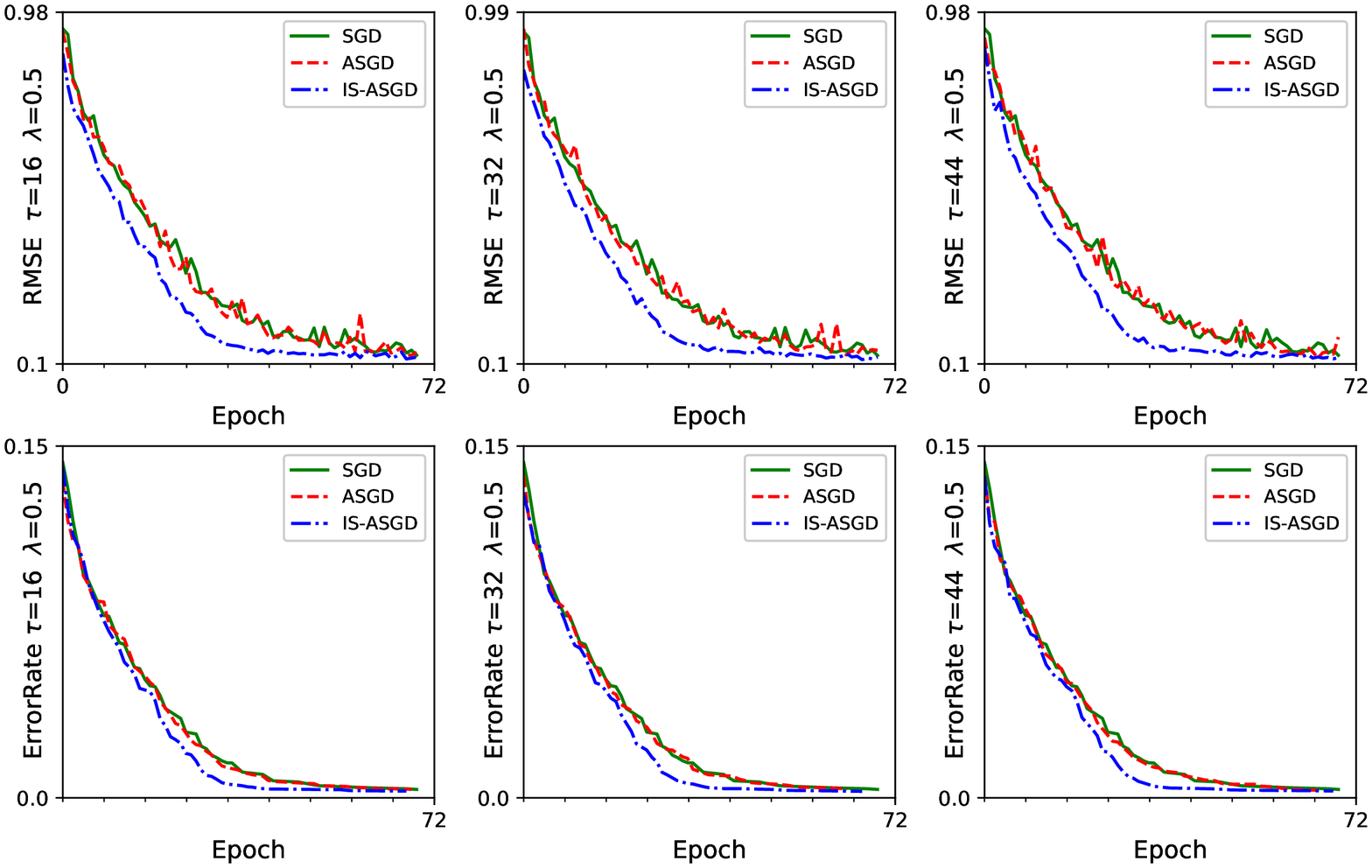}}%
	\caption{Iterative Convergence Result for SGD, ASGD, SVRG-ASGD and IS-ASGD}
	\label{convergence_result_iter}
\end{figure*}
\begin{table*}[]
	\centering
	\caption{Evaluation Datasets}
	\vskip 0.15in
	\begin{small}
		\begin{tabular}{ccccccc} \hline
			\hline
			\textbf{Name}&\textbf{Dimension}&\textbf{Instances}&\textbf{ $\nabla f_i$-Spa.}&\textbf{$\psi$}&\textbf{$\rho$}&\textbf{Source}\\ \hline
			News20 &1,355,191 & 19,996 &$10^{-3}$ & 0.972 &$5^{-4}$ &JMLR\\
			URL & 3,231,961 & 2,396,130 &$10^{-5}$ & 0.964 &$3^{-4}$ &ICML\\
			Algebra  &20,216,830  & 8,407,752  & $10^{-7}$ & 0.892 & $1^{-4}$ &KDD\\
			Bridge  &29,890,095 & 19,264,097 & $10^{-7}$ & 0.877 & $2^{-4}$ &KDD\\
			\hline
		\end{tabular}
	\end{small}
	\vskip -0.1in
\end{table*}
\begin{figure*}[ht]
	\centering
	\subcaptionbox{JMLR\_News20, $\lambda = 0.5$}{\includegraphics[scale=0.22]{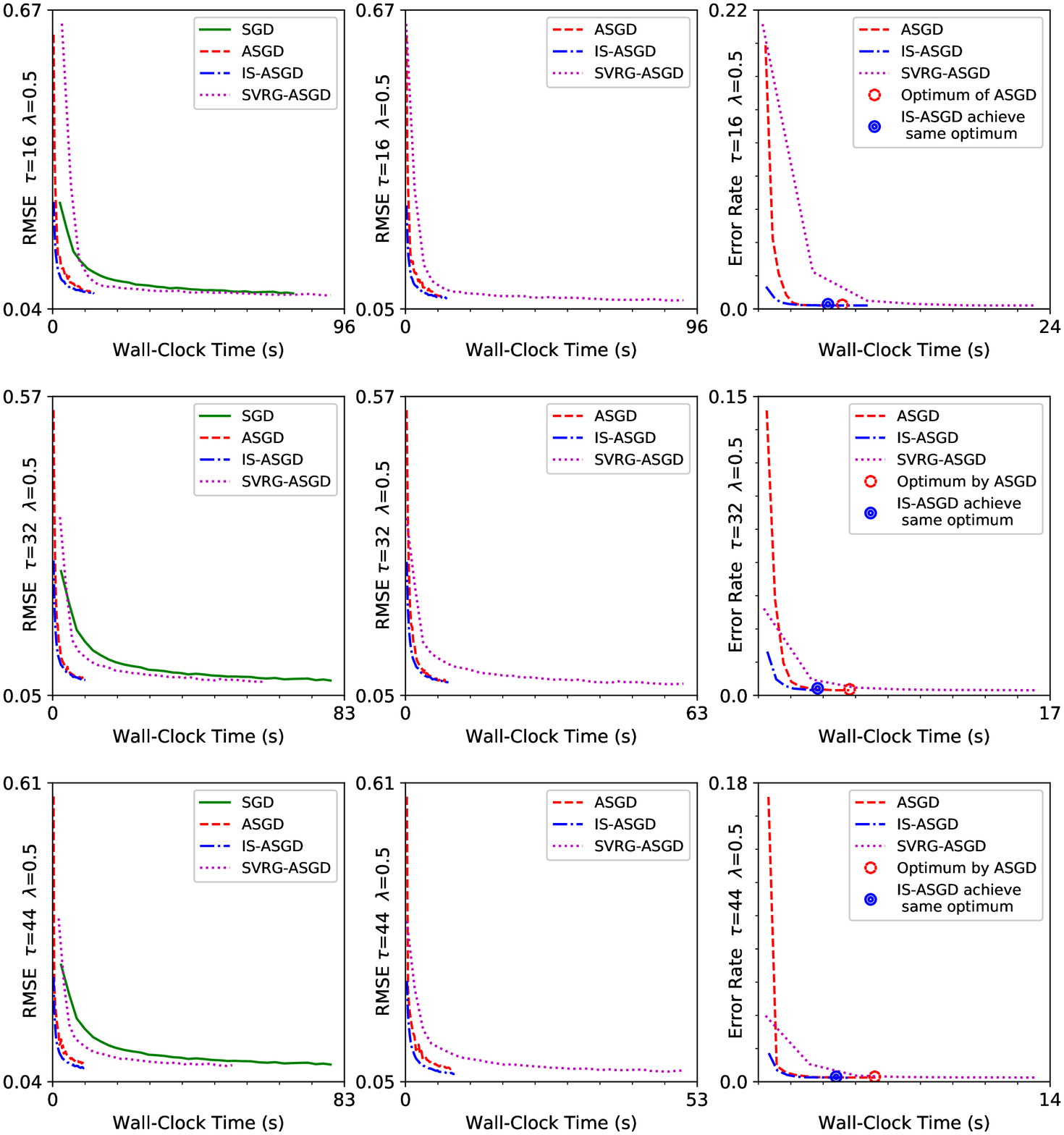}}%
	\subcaptionbox{KDD2010\_Alg., $\lambda = 0.5$}{\includegraphics[scale=0.22]{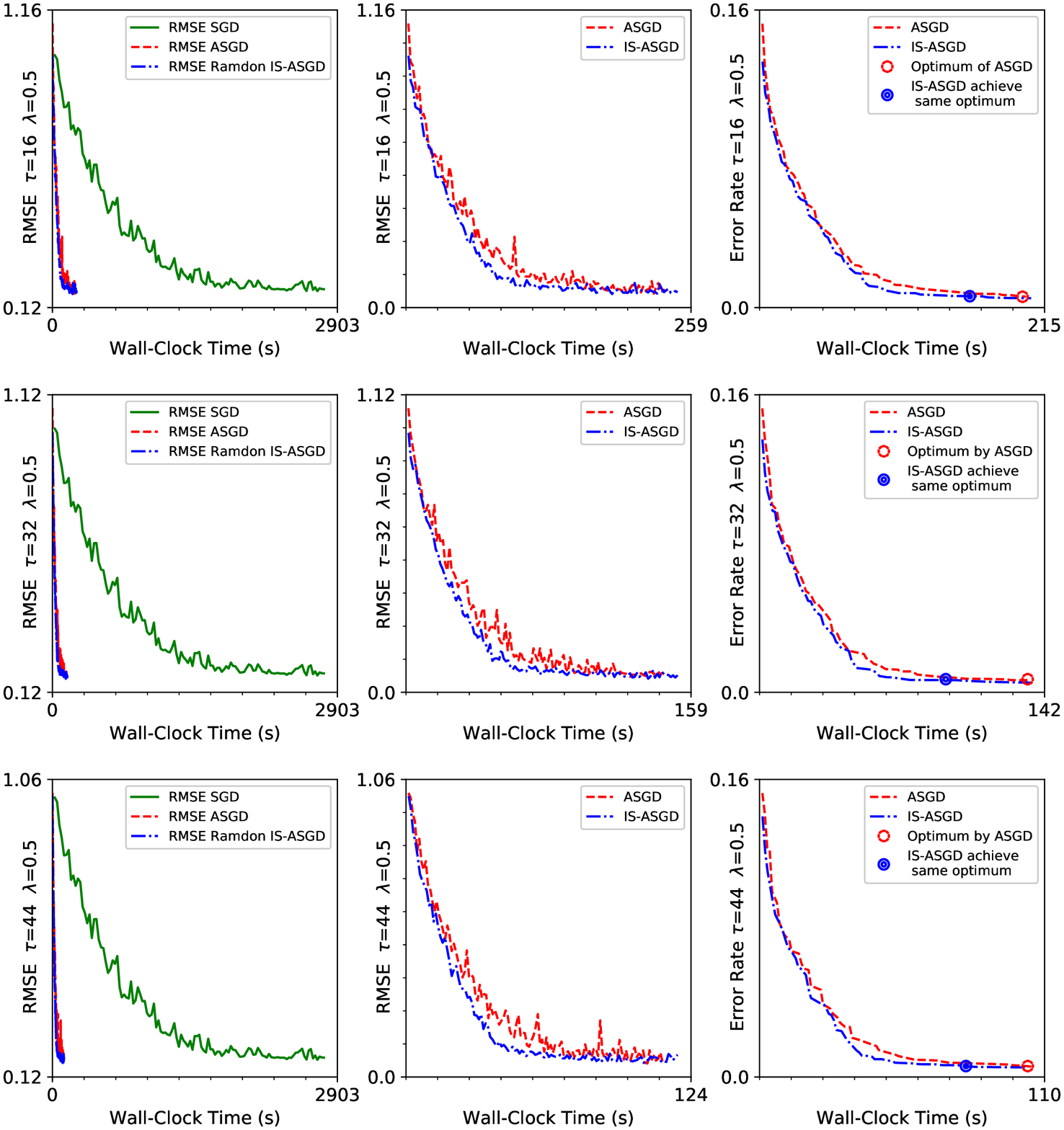}}%
	
	\subcaptionbox{ICML-URL, $\lambda = 0.05$}{\includegraphics[scale=0.22]{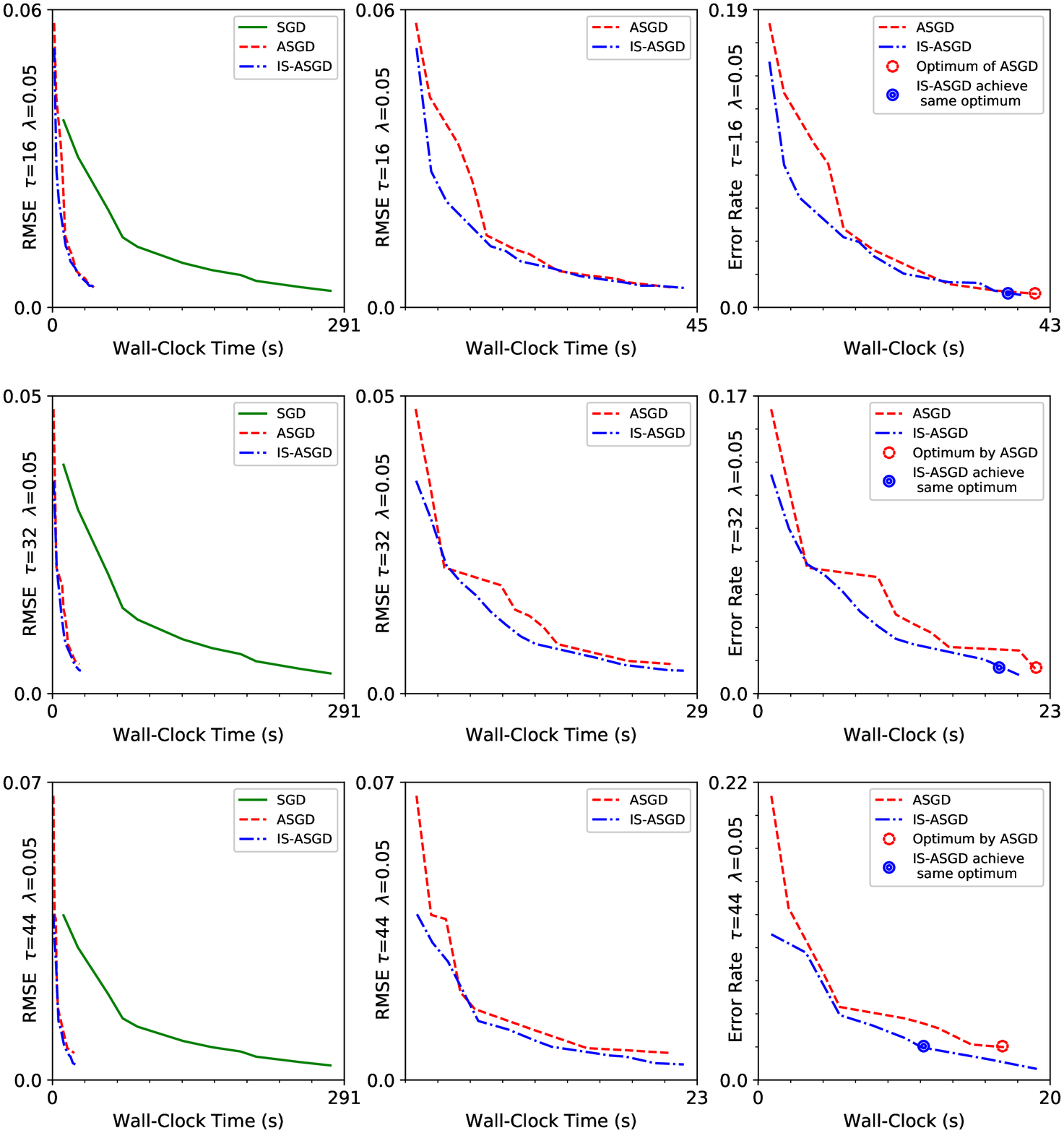}}%
	\subcaptionbox{KDD2010\_Bri., $\lambda = 0.5$}{\includegraphics[scale=0.22]{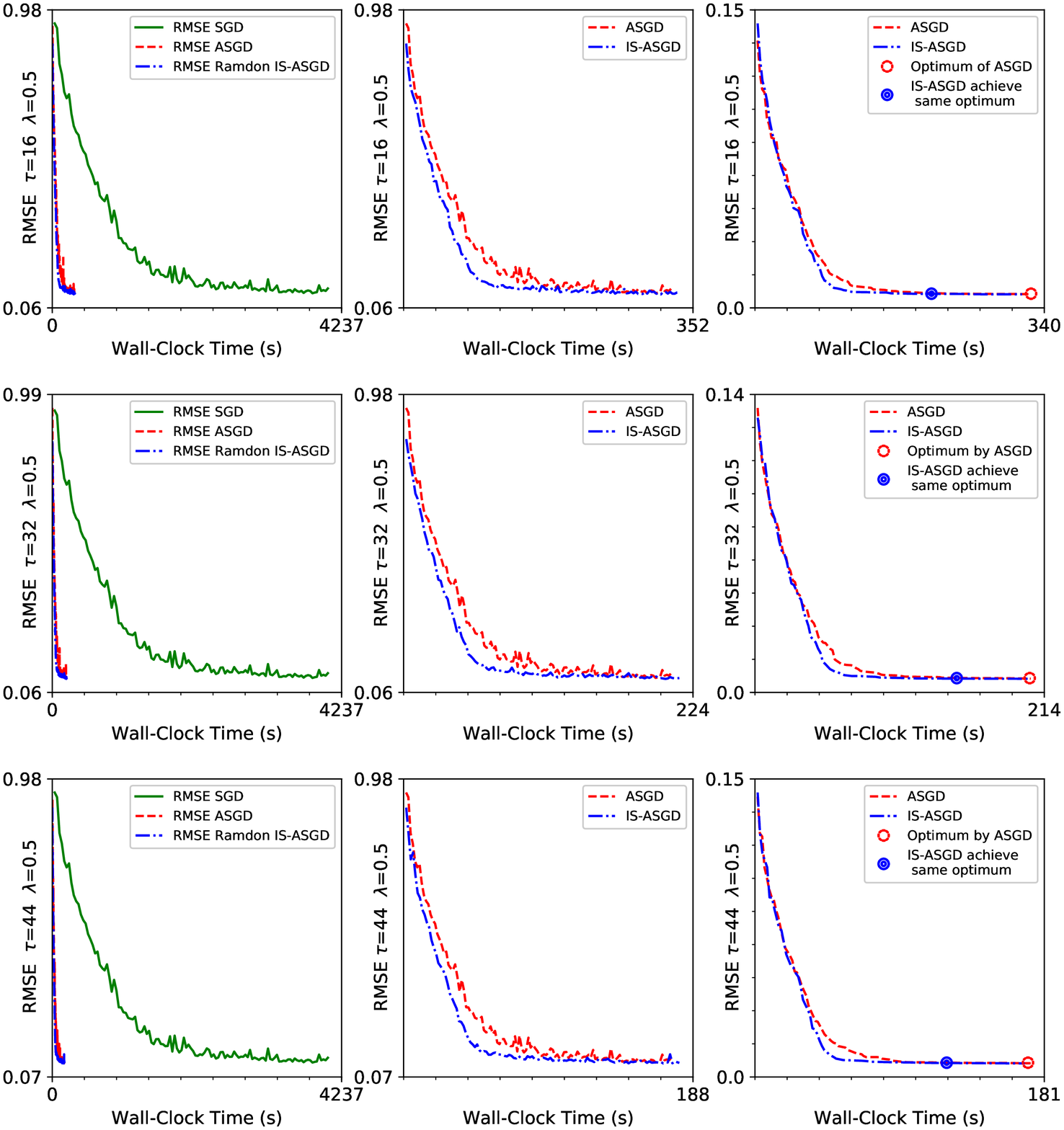}}
	\caption{Absolute Convergence Result for SGD, ASGD, SVRG-ASGD and IS-ASGD\\
	}
	\label{abs_convergence_result}
\end{figure*}
\noindent \textbf{Objective Functions}\hspace{0.2cm} We evaluate IS-ASGD based on the most widely used objective function in classification problems, i.e., L1-regularized cross-entropy loss. It has been popularly adopted from simple linear models to neural network based models.
\vskip 0.1cm
\noindent \textbf{Concurrency}\hspace{0.2cm}16, 32 and 44 threads are evaluated.
\vskip 0.1cm
\noindent \textbf{Algorithms}\hspace{0.2cm}Despite ASGD algorithm which is our target to accelerate, we also evaluate SGD (as baseline) and SVRG-ASGD.
\vskip 0.1cm
\noindent\textbf{SVRG-ASGD}\hspace{0.2cm}We implement SVRG-ASGD by strictly following the proposed algorithm in \cite{Reddi} without the skip-$\mu$ approximation in the available public version since this approximation deteriorates the convergence rate significantly.
\vskip 0.1cm
\noindent \textbf{Metrics}\hspace{0.2cm}Two metrics: rooted mean squared error (RMSE, objective value as the error) and error rate (i.e., misclassification) are evaluated and the error rate is updated once a better result is obtained.

\subsection{Iterative Convergence Rate Acceleration}
For iterative comparison, our expectation is that for same epoch counts, IS-ASGD achieves lower RMSE than ASGD and the error rate of IS-ASGD should be roughly (not strictly) lower than ASGD since a lower cost does not meant to be a lower error rate. Figure~\ref{convergence_result_iter} shows the iterative convergence results of all four datasets in their corresponding sub-figures respectively. 

\vspace{2px}
\noindent\textbf{Comparing to SVRG-ASGD}\quad From Figure~\ref{convergence_result_iter}-a, we see that SVRG-ASGD achieves the best iterative convergence rate with large improvement which comes at the price of magnitudes higher iterative computation cost. Meanwhile, it can also be noticed that with the increasing of the concurrency, the improvements of SVRG diminish quickly. This complies to our previous analysis that SVRG suffers higher potentiality of conflict updates due to its loss of sparsity, which makes it more concurrency-sensitive.

\vspace{2px}
\noindent\textbf{Comparing to ASGD}\quad We notice that the iterative convergence metrics of ASGD are the worst. It is worse than SGD in datasets that are relative dense, e.g., News20 and URL, while in datasets that are sufficiently sparse, e.g., the KDD datasets, its convergence rate are close to SGD. It is also clear that IS-ASGD's iterative convergence rate is much better than ASGD in all cases. In fact, IS-ASGD also achieves better optimum, i.e., a lower final error rate and RMSE as can be seen from the results.

They also show different concurrency-robustness, for instance, in Figure~\ref{convergence_result_iter}-c, when $\tau=16$, ASGD achieves close convergence rate to SGD with the increasing of epochs. However its convergence metrics deteriorates quickly when $\tau$ increases to 32 and 44. Meanwhile, IS-ASGD seems non-effected, it maintains close convergence results with SGD in all concurrencies which is a large improvement of ASGD and shows its concurrency-robustness.

Figure~\ref{convergence_result_iter}-b and d shows similar results that IS-ASGD achieves significant convergence rate accelerations comparing to ASGD and SGD. These two datasets, i.e., KDD2010\_Alg., and Bri. are sparse and have extremely large dimensionality and number of samples. They also have lower $\psi$, as mentioned in Equation 15, section 2.2, the convergence bound improvement of applying IS in SGD is negatively correlated to $\psi$ ant thus IS-ASGD achieves much significant convergence improvements in these two datasets, which is even much better than SGD. While in Figure ~\ref{convergence_result_iter}-a and Figure ~\ref{convergence_result_iter}-c where the datasets are relatively small (potentially higher imbalance), dense and have higher $\psi$, its convergence bound are close to SGD. 

In fact, when conditions discussed in section 3 are satisfied, i.e., datasets are sufficiently sparse, the iterative convergence rate of IS-ASGD will be no worse than SGD while the iterative convergence rate of ASGD will be no better than SGD. When datasets are even more large-scale and has lower $\psi$, the convergence rate improvement of IS-ASGD increases significantly. We can firmly say that IS-ASGD accelerates the iterative convergence rate of ASGD effectively due to its inherited superior convergence bound from IS-SGD. Such improvements will directly result in absolute convergence rate acceleration of IS-ASGD since its iterative time cost remains almost the same with ASGD, and most importantly, it preserves the sparsity.


\begin{figure*}[]
	\centering
	\subcaptionbox{JMLR\_News20, $\lambda = 0.5$}{\includegraphics[scale=0.2]{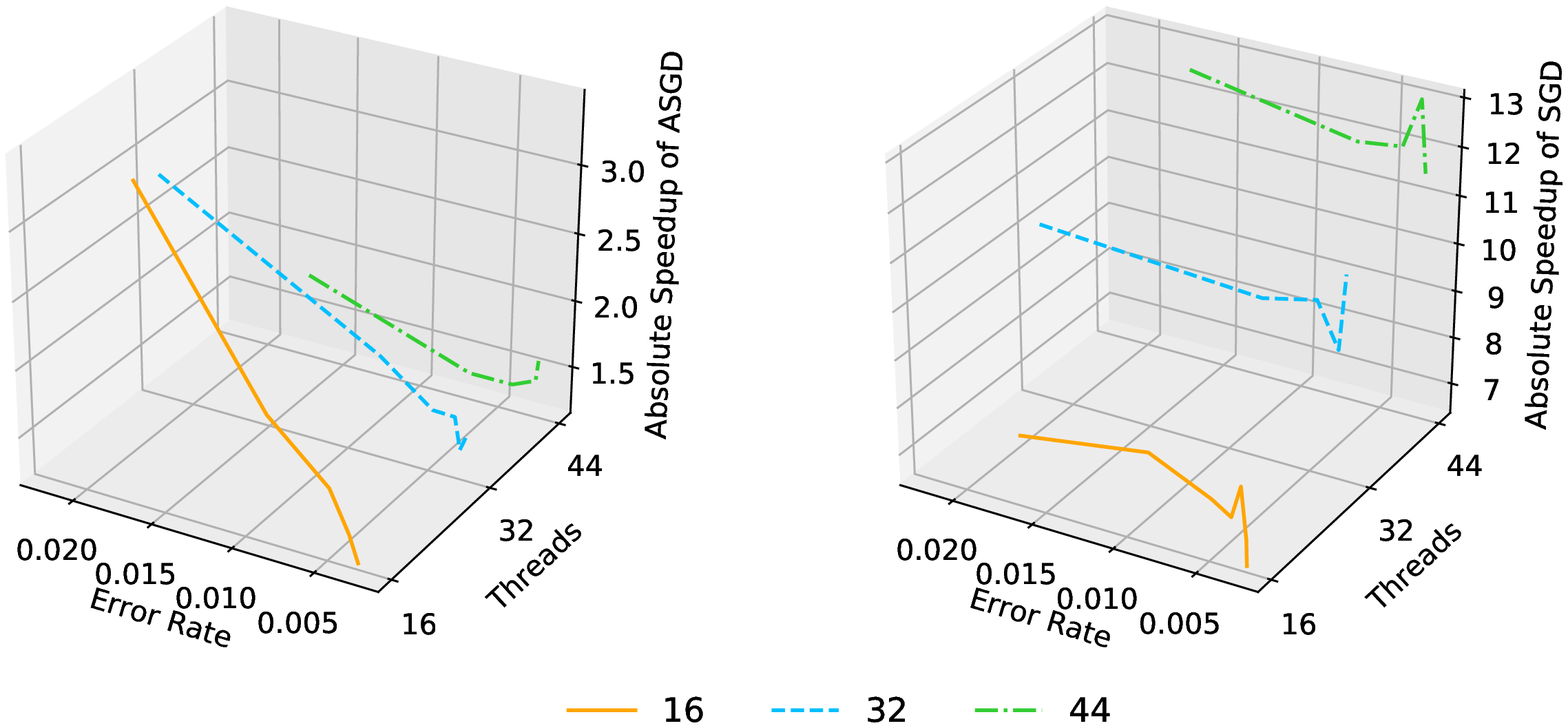}}%
	\subcaptionbox{KDD2010\_Algebra, $\lambda = 0.5$}{\includegraphics[scale=0.2]{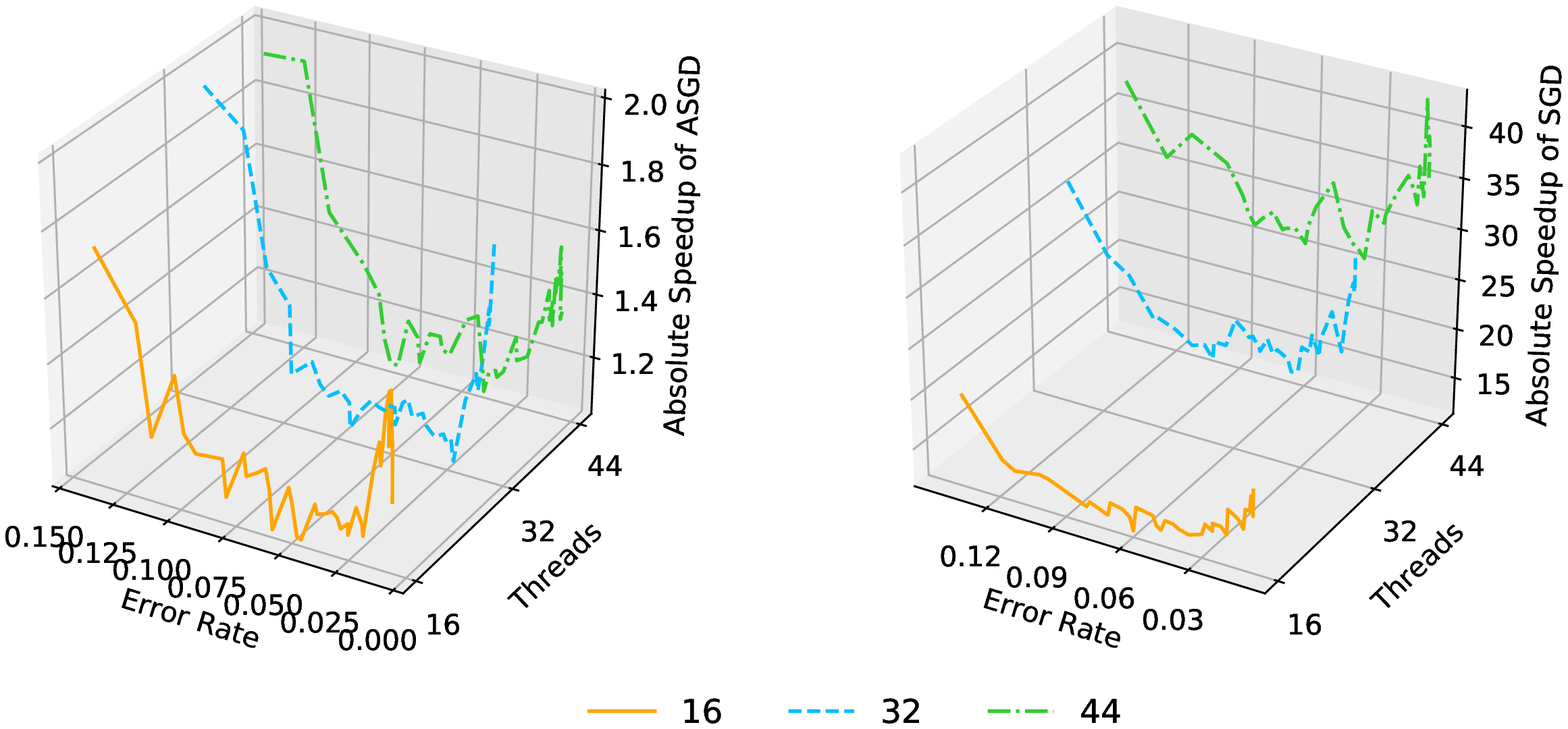}}%
	
	\subcaptionbox{ICML\_URL, $\lambda = 0.05$}{\includegraphics[scale=0.2]{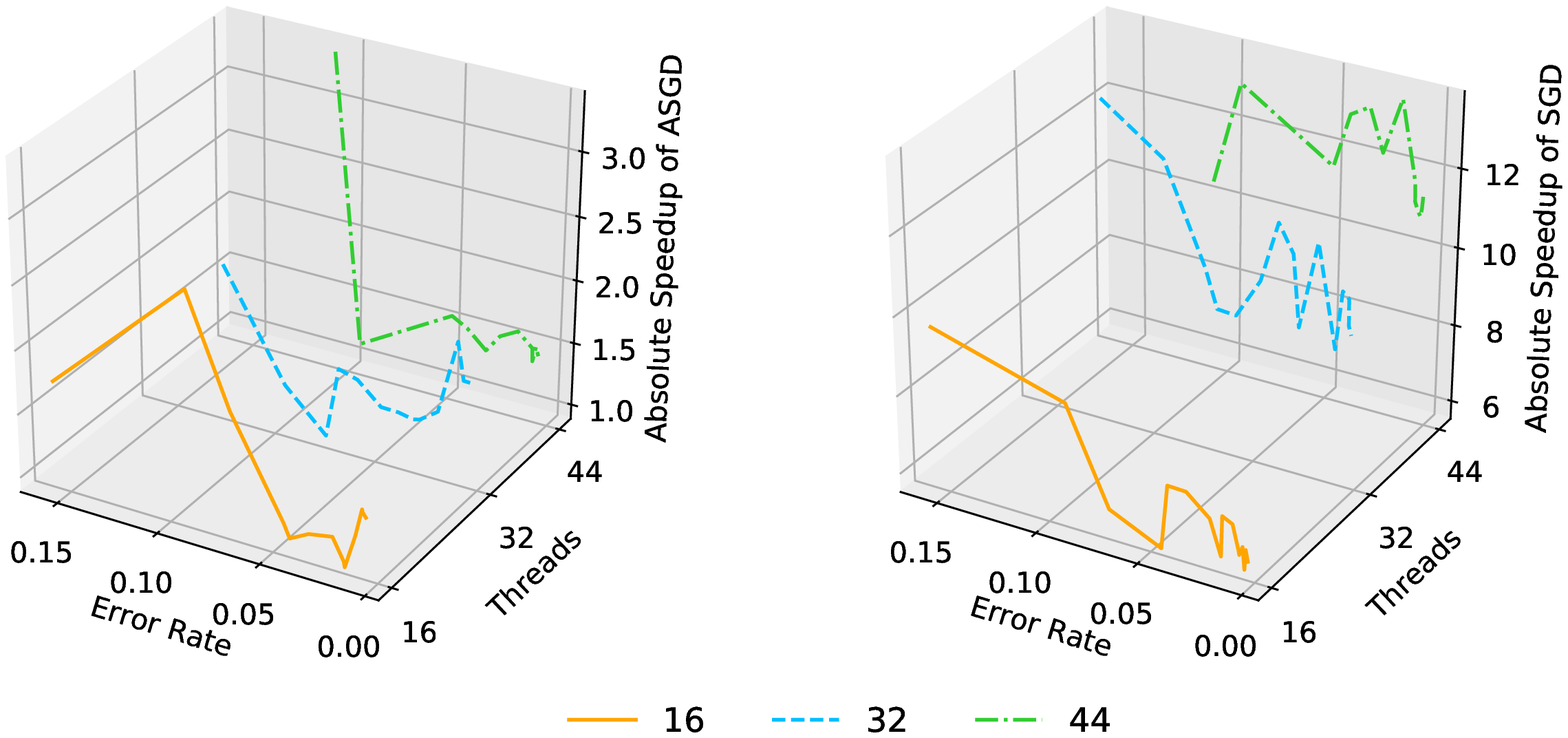}}%
	\subcaptionbox{KDD2010\_Bridge, $\lambda = 0.5$}{\includegraphics[scale=0.2]{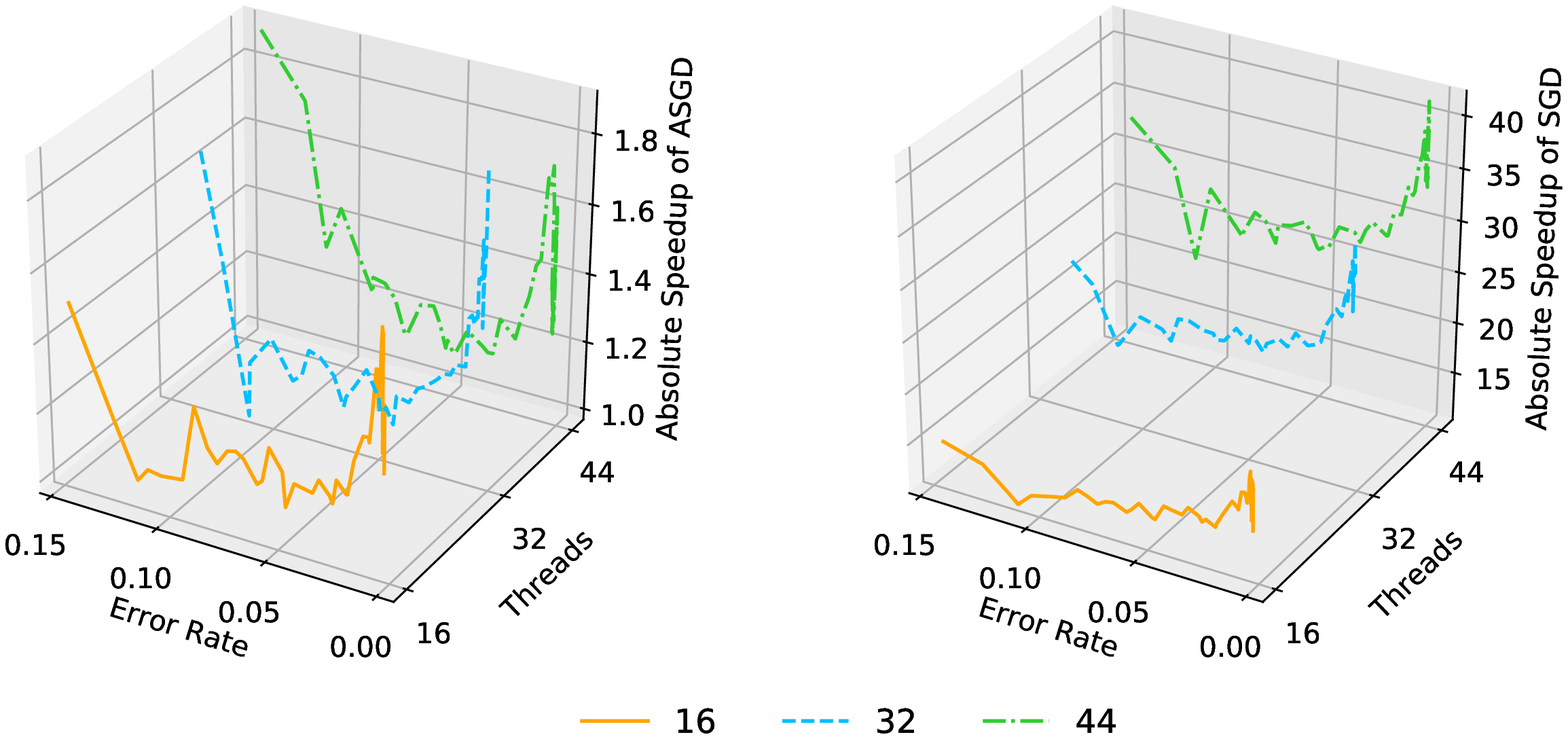}}
	\caption{ErrorRate - Speedup of IS-ASGD}
	\label{real}
\end{figure*}

\subsection{Absolute Convergence Rate Acceleration}
We present the results of absolute convergence acceleration in two forms as shown in Figure~\ref{abs_convergence_result} and Figure~\ref{real} respectively. While the iterative convergence results shown above hold a more academic meaning, the absolute convergence rate is the metric that matters for practical deployments since people always hope to obtain a trained model with less time. Figure \ref{abs_convergence_result} plots the absolute convergence curves with the x-axis as wall-clock in seconds. Be noted that we provide the RMSE comparison between all algorithms in the first column while in the second and third columns only the RMSE and error rate comparison between ASGD and IS-ASGD are shown for a better resolution since their curves are very short comparing to that of SGD and SVRG-ASGD.

As can be seen from Figure~\ref{abs_convergence_result}-a, SVRG-ASGD takes much longer time to achieve the same accuracy than other algorithms despite of its superior iterative convergence rate as shown in Figure~\ref{convergence_result_iter}-a since its iterative computation cost is several magnitudes higher than others. For the comparison between ASGD and IS-ASGD, we specifically plot the final best error rate (referred to as optimum) of ASGD in red circle while the blue dot corresponds to the same optimum achieved by IS-ASGD. This comparison directly shows the final absolute speedup results of IS for ASGD. We see that in Figure~\ref{abs_convergence_result}-d, IS-ASGD achieves a maximum 1.8x acceleration while in other cases its acceleration of the optimum varies depending the datasets and concurrency. In Figure~\ref{abs_convergence_result}-c, it shows that IS-ASGD also achieves better final optimum and higher acceleration of ASGD when concurrency increases.

The results also show that the optimums of error rate are achieved much earlier than that of RMSEs which implies that the acceleration of the early stage of the convergence is more important since for later stage the error rate improvements is very limited. We thus present error-rate/speedup curves for an in-depth inspection.

\vspace{2px}
\noindent\textbf{In-depth: Slice Inspection}\quad Figure~\ref{real} is derived from Figure~\ref{abs_convergence_result} directly, the two 3D-figures in each sub-figure show the speedups of IS-ASGD over ASGD and SGD in a slicing manner for a deeper inspection of the convergence procedure. Its y-axis is the concurrency and z-axis is the absolute speedup of reaching the corresponding error-rate (values are linearly interpolated when needed) on x-axis. 

From the speedup curves, we can see that the speedups are the largest at the early stage and drop in the middle. We can also conclude that the scales of the datasets affect the speedup curves in two aspects, first, for large-scale datasets, i.e, in Figure~\ref{real}-b and d, the speedups rise at the final stage of the convergence procedures which implies that IS-ASGD achieves its best acceleration performance in large-scale datasets when searching for the optimal models. Second, for large-scale datasets, the average speedups of IS-ASGD over ASGD seem to be invariant to the currency as the curves show similar shape and mean, which indicates concurrency-robustness.

As can be summarized, the average speedups of IS-ASGD over ASGD range from 1.26 to 1.97 while the optimum speedups range from 1.13 to 1.54. For the raw computational speedup, it can be seen that the speedups of IS-ASGD over SGD for 16 threads range from 6.39 to 12.29 and increase to 11.89 to 23.53 when threads count increases to 44 depending on the size of dataset. In general, small data size does not achieve a good raw computation speedup. Taking the sampling time into consideration, the raw computational speedups of IS-ASGD are typically 7.7\% to 1.1\% lower than ASGD which are relative small differences. If we generate the sample sequence of IS-ASGD for each thread only once and simply shuffle it every epoch, there will be no computation performance gap between ASGD and IS-ASGD. In fact, such approximation work well in practice according to our evaluation.
\subsection{Discussion: When Datasets are Dense}
The main reason that causes SVRG performs inefficiently in large-scale sparse datasets is its reliance on the dense gradient $\mu$ which is magnitudes larger than the stochastic gradient $\nabla f_i$. On the other hand, if the datasets are dense, e.g., when the sparsity of $\nabla f_i$ is higher than $10^{-3}$ which is close to $\mu$, SVRG-ASGD prevails since their iterative computation costs are in same magnitude and SVRG-ASGD's iterative convergence rate is much higher. Additionally, when datasets are small-scale, the whole training procedure tends to finish quickly. For this case, the performance bottleneck is the overhead of multi-process scheduling, all reduce operation, etc., instead of the computation, and thus SVRG-ASGD is likely to outperform the other algorithms. Since small datasets are of only academic meanings, it seems that the proper applications for SVRG-ASGDs are the scenarios when using ASGD for relative dense datasets. However for most large-scale optimizations, the datasets are typically sparse with its sparsity significantly lower than $10^{-3}$.
\section{Conclusion}
Techniques for the acceleration of the convergence rate of asynchronous stochastic optimizations are of great importance and has long been a hot research field. In this paper we located several unidentified bottlenecking issues for current SVRG-based ASGD acceleration algorithms in large-scale sparse datasets and propose the novel IS-ASGD algorithm which avoids the above bottlenecking issues naturally. Its key advantage lies in the capability of accelerating the iterative convergence rate of ASGD with few increasing of the iterative time cost which in turn results in effective acceleration of the absolute convergence rate. We use importance-balancing trick to balance the importance between asynchronous cores/nodes which helps preserving the optimal VR performance of IS. Moreover, we theoretically proved the convergence bound of IS-ASGD which shows that IS-ASGD speeds up IS-SGD almost linearly and consequently inherits its superior convergence bound over ASGD and SGD. The experimental evaluation results clearly verify that IS-ASGD achieves 1.13~1.54x absolute convergence rates acceleration for ASGD. Evaluation codes can be publicly accessed.
\bibliographystyle{ACM-Reference-Format}
\bibliography{final} 
\end{document}